\documentclass[10pt, conference]{styles/IEEEtran}

\usepackage{epsfig}
\usepackage{amssymb} 
\usepackage{amsthm}
\usepackage[tbtags]{amsmath} 
\usepackage{graphics,eepic,epic}
\usepackage{latexsym}
\usepackage{euscript}
\usepackage{subfigure}
\usepackage{graphics,eepic,epic,psfrag}
\usepackage{flushend}

\newcommand{\bbE}{\mathbb{E}}

\newtheorem{theorem}{Theorem}

\newtheorem{definition}{Definition}
\newtheorem{lemma}[theorem]{Lemma}

\usepackage{cite}

\begin{document}

\title{\begin{huge}
Queuing Theoretic Analysis of Power-Performance Tradeoff in Power-Efficient Computing  \end{huge}}
\pagenumbering{arabic}

\author{
\authorblockN{Yanpei Liu, Stark C.~Draper, Nam Sung Kim}
\authorblockA{Electrical and Computer Engineering, University of Wisconsin Madison\\
Email: \{yliu73@, sdraper@ece, nskim3@\}wisc.edu}
}


\maketitle
\pagenumbering{arabic}

\begin{abstract}
In this paper we study the power-performance relationship of power-efficient computing from a queuing theoretic perspective. We investigate the interplay of several system operations including processing speed, system on/off decisions, and server farm size. We identify that there are oftentimes ``sweet spots" in power-efficient operations: there exist optimal combinations of processing speed and system settings that maximize power efficiency. For the single server case, a widely deployed threshold mechanism is studied. We show that there exist optimal processing speed and threshold value pairs that minimize the power consumption. This holds for the threshold mechanism with job batching. For the multi-server case, it is shown that there exist best processing speed and server farm size combinations. \\
\end{abstract}


\section{Introduction}

Large-scale data center networks have gained tremendous usage nowadays. Applications running inside such clustered severs include web searching, e-commerce, and compute-intensive applications. However, today's data centers spend a large amount of capital on power usage and other associated infrastructures. Around $40\%$ of total operation cost is related to power distribution, cooling and electricity bills \cite{CostOfCloud}. In $2005$, the total data center power consumption was $1\%$ of the total U.S.~power consumption and caused emissions as much as a mid-sized country such as Argentina \cite{MathewSitaraman}. Emphasizing the importance of these issues, we note that recently the U.S.~Environment Protection Agency raised concerns to the Congress about the growing power consumption in data centers \cite{USEPA}. 

Much power consumed by data centers is wasted: servers on average are only $10-50\%$ utilized \cite{CostOfCloud, VermaAhuja, MeisnerGold, BodikArmbrust}. Low utilization is epidemic to data center operations due to strict service level agreements on peak workload provisioning. However, due to the lack of ``power proportionality", an idling server still consumes $60\%$ of its peak power, drawn mainly in peripherals such as DRAM, hard disk drivers (HDDs), network interface card (NIC), etc. Thus, to conserve power it is preferable to shut down servers. When considering server farms consisting of multiple servers, jobs can be consolidated into a few servers so that the rest can be shut down. Server on/off decisions are often made in conjunction with processing speed adjustments. Dynamic voltage and frequency scaling (DVFS) is a conventional processing speed adjustment technique that changes the processor's clock frequency (and thus the speed of computation) according to workload conditions in order to reduce power consumption. Server on/off decisions (also known as dynamic component shut-down) and processor speed adjustment (also known as dynamic performance scaling) can be categorized in Figure~\ref{fig.category} (see Figure 5 in \cite{BeloglazovBuyyaLee} for a complete diagram).

\begin{figure}[t]
\includegraphics[scale=0.5]{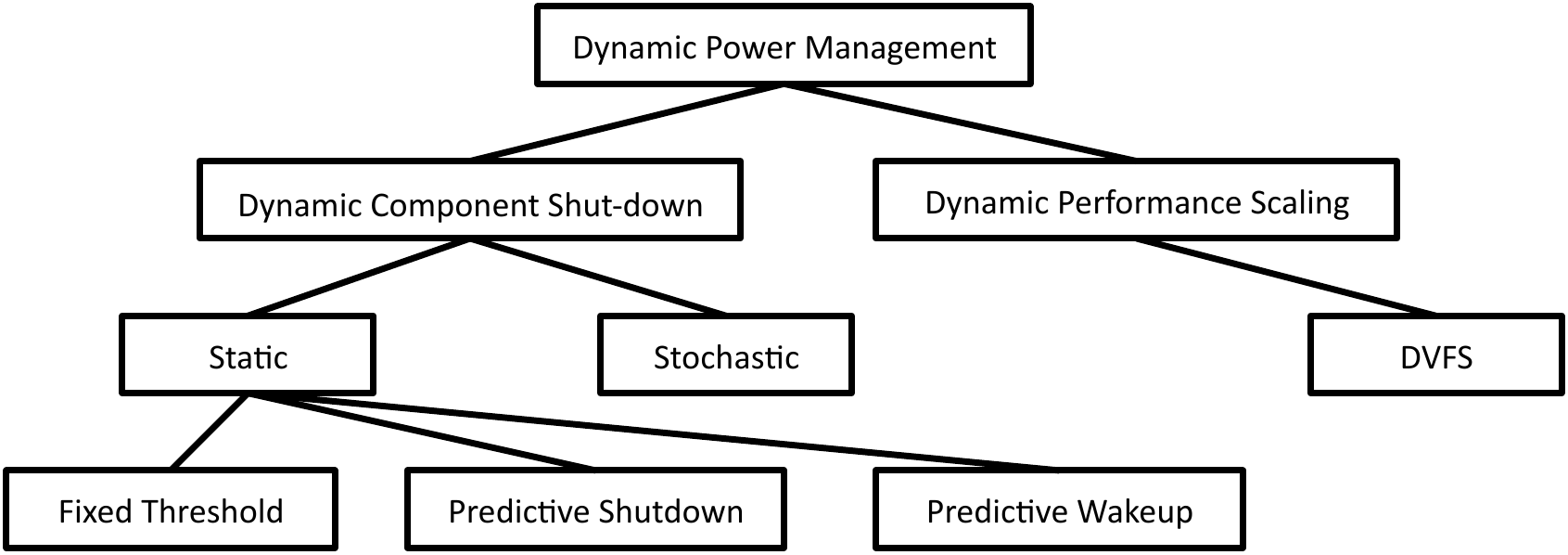}
\caption{Dynamic power management category \cite{BeloglazovBuyyaLee}}
\label{fig.category}
\end{figure}

Our results tie into many earlier works, both in the computer architecture and queuing theoretic communities. The authors in \cite{MeisnerGold} study a power saving method that shuts down servers when they are in idle and characterize the power-delay tradeoff from a queuing theoretic perspective. However they do not consider the performance scaling in their theoretic analysis. The authors in \cite{MeisnerSadler} investigate power reduction possibilities for jobs that demand fast response.  They suggest that system-wide coordinated power management provides a far better power-latency tradeoff than individual uncoordinated decisions. The work in \cite{LangPatel} also make similar statement. The authors study power management for MapReduce tasks, suggesting that all nodes in a MapReduce cluster should be powered up and down together rather than individually in a distributed fashion. The authors in \cite{MedanBuyuktosunoglu} highlight the challenges of avoiding negative power saving. Negative power savings occur when the overhead of implementing the power-savings mechanism exceeds the resulting savings, thus costing the system extra power. They suggest guard mechanisms to monitor negative power savings and performance degradation caused by those power saving routines. The impact of data center size on power efficiency is evaluated in \cite{GandhiHarchol}. Most of the above works consider variants of a fixed threshold mechanism.  In such mechanisms a server is shut down whenever it exceeds some idleness threshold. Stochastic on/off decisions are studied in \cite{Neely} and stochastic optimization methods are also used in \cite{YaoHuang, NeelyTehrani}. For other related works on predictive shut-down and wake-up, see \cite{BeloglazovBuyyaLee}.


Surprisingly, although component shut-down and performance scaling are widely used mechanisms in power-efficient computing, little is known from the queuing theoretic perspective, especially when component shut-down is jointly considered with performance scaling. The power-performance tradeoff in these settings is not well understood.  This often results in suboptimal designs. We aim to study the fundamental interplay between these system operations including processing speed, on/off decisions, and server farm size from a queuing theoretic point of view. Our results yield clear design guidance.  One result demonstrates that there are sweet spots in power-efficient computing.  These are optimal processing speed in combination with various other system parameter settings that yield the greatest power savings.  Somewhat surprisings these results contrast to much conventional wisdom that underlies many protocols such as the ``race-to-halt" mechanism.  Race-to-halt suggests that one run the processor as fast as possible and then shut it down.  In contrast, the sweet spots we identify show that it can be more power-efficient (for a given computational performance target) to run the processor more slowly for longer.  To develop these results, in this paper we first study the interplay between fixed-threshold reactive power control mechanism and DVFS to identify the optimal operation settings. The optimal settings also appear in the threshold mechanism with job batching, i.e., batching certain amount of jobs before system wake-up. We then extend the concept to the multi-server case where we consider the relation between server farm size and processing speed.  

The rest of the paper is organized as follows. In Section~\ref{sec.server_model} we present the server model. In Section~\ref{sec.singleserver} we present the analysis for the single server case. The muti-server case is discussed in Section~\ref{sec.multiserver}. We conclude in Section~\ref{sec.conclusion}.

\section{Server Model}
\label{sec.server_model}
We model each server as a computation entity that processes jobs. Each server is equipped with a DVFS mechanism. DVFS is a conventional method widely used to trade off power consumption with processing speed by changing operating voltage and clock frequency. We assume the clock frequency can be scaled by a factor $f \in [0, 1]$ and the time it takes to process each job under DVFS is exponentially distributed with mean $1/\mu f$. For simplicity, herein we assume the processing time for all jobs is independent and identically distributed. Setting $f = 1$ yields maximum processing speed $1/\mu$ and setting $f = 0$ stops the server from processing jobs, i.e., the server is in the clock-gated mode.

The dynamic power consumption of a system supporting DVFS is proportional to $ V^2 f $ where $V$ is the supply voltage and $f$ is the clock frequency scaling factor. The supply voltage is determined by frequency and can be reduced if the clock frequency is also reduced. This results in a cubic reduction in power consumption. Therefore we model the power consumed by the server as $P_0f^3 + C$ where $P_0$ is the maximum power draw from the computing entity itself, e.g.~CPU. The second term $C$ is the average power drawn by peripherals such as DRAM, hard disk drivers (HDDs), network interface card (NIC), etc. This can be thought of as the ``infrastructure" cost incurred by keeping the computational unit on and ready to process jobs. Note that when $f = 0$, i.e., the server is in the clock-gated mode: the power consumed by the server is the peripheral power $C$. This is different from the mode that the server is shut-down in which case the power consumption is zero. When the server is shut-down, there is a wake-up penalty in terms of time and power. For the ease of illustration, we model the peripheral power $C$ as independent of $f$. In practice, the peripheral power also depends on the system operation, exhibiting different values in active, idle and sleep modes \cite{MeisnerGold}.

\section{Single Server Analysis}
\label{sec.singleserver}
In this section we provide our analysis for the single server case. We assume jobs arrive according to a Poisson process with arrival rate $\lambda$. We study two conventional power-saving operations, namely the threshold mechanisms with and without job batching. 



\subsection{Threshold mechanism}
\label{sec.singleserver_threshold}
We first describe the mechanism without job batching. 
\begin{definition}[Threshold Mechanism]
\label{def.threshold_mechanism}
The server processes jobs until the queue is empty. Then it waits for a fixed amount of time threshold $\tau_c$. If the next job arrives {\em within} this waiting threshold $\tau_c$, the server processes the job and resumes normal operation. Otherwise the server shuts down in which the whole platform (CPU and the peripherals) is powered down consuming zero power. If the next job arrives {\em after} the waiting threshold $\tau_c$ (thus after the server has powered off), the server takes time $\tau_s$ to wake up before processing the job. 
\end{definition}
A pair of sample paths illustrating the operation of this mechanism is provided in Figure~\ref{fig.threshold_illustration}.  The upper sample path indicate queue occupancy.  The lower sample path is binary, indicating when the server is on and off. 
\begin{figure}[t]
\centering
\includegraphics[scale=0.45]{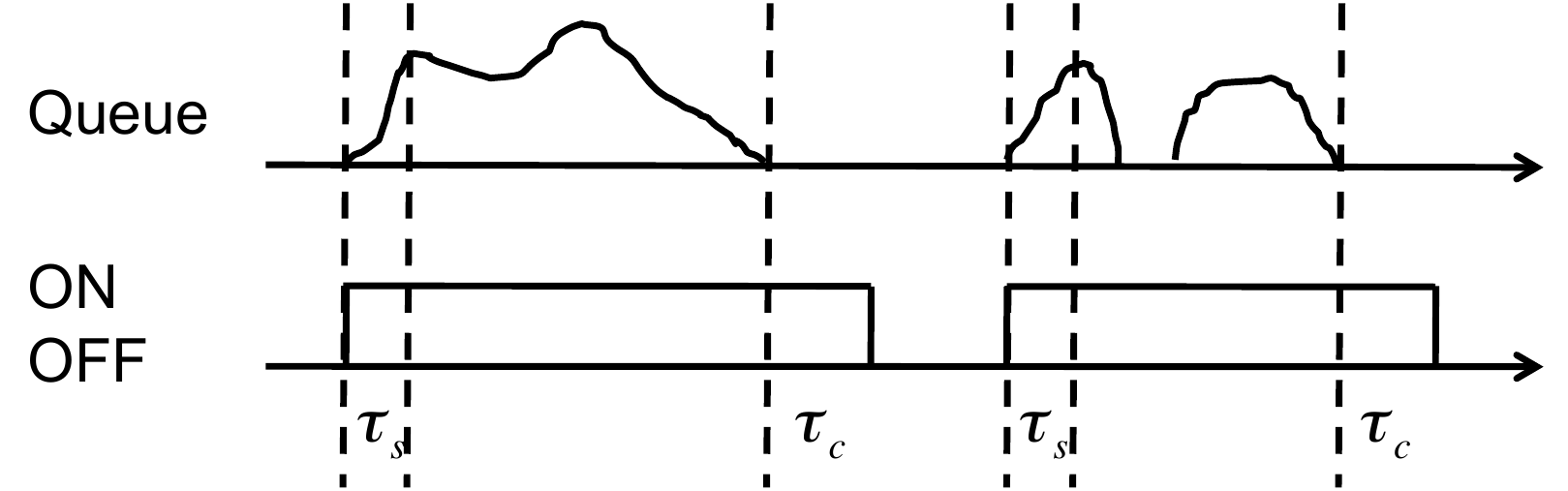}
\caption{Threshold mechanism.}
\label{fig.threshold_illustration}
\end{figure}
For the ease of illustration, we assume whenever the server is not shut-down, its power consumption is consistent over time determined by the frequency scaling $f$. Our analysis can be easily extended to the case where the power spent in $\tau_c$ and $\tau_s$ are different from the normal operation. 

Surprisingly for such a widely used mechanism, to the best of our knowledge, it has not been thoroughly studied from the queuing theoretic prospective. Indeed, it is not immediately clear how mean response time and power consumption are related under frequency scaling $f$ and peripheral power $C$. In current implementation, the threshold value $\tau_c$ is chosen as a fixed value mostly based on operators' own experience \cite{MedanBuyuktosunoglu}. We investigate how the waiting threshold $\tau_c$, frequency scaling $f$ and wake-up latency $\tau_s$ jointly affect the power and mean response time of such systems. We study this via a queuing theoretic analysis. Our results reveal that it is important to determine these operation parameters in a joint fashion. Na\"{i}vely picking $\tau_c$ too large or too small may lead to poor power efficiency. 

The following theorem summarizes the relationship between mean response time $\bbE[R]$ and power consumption $\bbE[P]$.
\begin{theorem}
\label{thm.threshold}
The mean response time and mean power consumption of a server using the threshold mechanism are given by:
\begin{align}
\label{eq.threshold_R}
\bbE[R] &= \frac{1}{\mu f - \lambda} + \frac{2\tau_s + \lambda \tau_s^2}{2(e^{\lambda \tau_c} + \lambda \tau_s)} \\
\label{eq.threshold_P}
\bbE[P] &= (P_0 f^3 + C) \left ( 1 - \frac{1 - \frac{\lambda}{\mu f}}{e^{\lambda \tau_c} + \lambda \tau_s } \right ).
\end{align}
\end{theorem}
\begin{proof}
It is shown that the mean response time for an M/G/1 queue with the first customer experiencing a random delay $D$ is given by \cite{Welch}:
\begin{align}
\label{eq.MG1_exceptional}
\bbE[R] = \frac{1}{f \mu} + \frac{\lambda (1 + c_s^2)}{2 f^2 \mu^2 \left ( 1 - \frac{\lambda}{f \mu} \right)} + \frac{2 \bbE[D] + \lambda \bbE[D^2] }{2(1 + \lambda \bbE[D])},
\end{align}
where $c_s^2$ is the variance of coefficient. The random delay $D$ in our case is $D = 0$ if $0 \leq T \leq \tau_c$ and $D = \tau_s$ if $T > \tau_c$ where $T$ is the time elapse to see the first arrival after the server runs out of jobs. The random variable $T$ is exponentially distributed with parameter $\lambda$. Therefore $\bbE[D]$ can be calculated as:
\begin{align}
\bbE[D] = \int_{\tau_c}^{\infty} \tau_s \lambda e^{-\lambda t} dt = \tau_s e^{-\lambda \tau_c}.
\end{align}
Similarly, $\bbE[D^2] = \tau_s^2 e^{-\lambda \tau_c}$. Plugging them into (\ref{eq.MG1_exceptional}) with $c_s^2 = 1$ for M/M/1 we obtain the mean response time (\ref{eq.threshold_R}).

The power expression can be derived as follows. Note that 
\begin{align}
\label{eq.6}
\bbE[P] = (P_0f^3 + C) (1 - f_{\rm off}),
\end{align}
where $f_{\rm off}$ is the fraction of the time the server is off. Now consider a time duration $L$ from the end of one epoch that the queue is empty to the end of next epoch that the queue is empty. Since this time duration starts with zero job and ends with zero job in the queue, the following equality holds:
\begin{align}
\lambda L = \mu f \left(L - \frac{1}{\lambda} - \bbE[D]\right).
\end{align}
Within this time duration, the server will shut down only when next job arrives after $\tau_c$. Thus $f_{\rm off}$ can be calculated as:
\begin{align}
f_{\rm off} = \frac{\int_{\tau_c}^{\infty} (t - \tau_c) \lambda e^{-\lambda t} dt}{L}.
\end{align} 
Plugging in $f_{\rm off}$ into (\ref{eq.6}) we obtain the power consumption (\ref{eq.threshold_P}).
\end{proof}

From Theorem~\ref{thm.threshold}, we have the following observations. First when $\tau_c = \infty$, the server never shuts down. The mean response time (\ref{eq.threshold_R}) and power consumption (\ref{eq.threshold_P}) reduce to:
\begin{align}
\bbE[R] = \frac{1}{\mu f - \lambda} \quad \quad  \bbE[P] = (P_0 f^3 + C), 
\end{align}
which is the mean response time and power consumption for an M/M/1 queue with frequency scaling $f$. 

When $\tau_s = 0$, i.e., the server incurs no delay to wake up. The mean response time reduces to an M/M/1 case while the power consumption can be minimized by picking $\tau_c = 0$. Thus we have:
\begin{align}
\label{eq.RP_Ts0}
\bbE[R] = \frac{1}{\mu f - \lambda} \quad \quad \bbE[P] = \frac{\lambda}{\mu f} (P_0 f^3 + C).
\end{align}
This means that if there is no cost to wake up a server, the server should shut down immediately when the queue becomes empty. However, note that the power-delay tradeoff is not monotonic: there is an optimal frequency that minimizes the power consumption (c.f.~Figure~\ref{fig.Ts0Tc0}). In other words, it is not always the case that running slow (while incurring large delay) leads to more power savings.

For a fixed nonzero $\tau_s$, there is an optimal $(\tau_c, f)$ pair that minimizes the power consumption for a given delay performance. To see this, fix $\bbE[R] = R'$ and from (\ref{eq.threshold_R}) we obtain the relationship between $f$ and $\tau_c$:
\begin{align}
\frac{1}{e^{\lambda \tau_c} + \lambda \tau_s} = \frac{2}{2\tau_s + \lambda \tau_s^2} \left(R' - \frac{1}{\mu f - \lambda}\right).
\end{align}
Plugging it into (\ref{eq.threshold_P}), we see that the optimal frequency scaling $f$ is the one that minimizes the following:
\begin{align}
\label{eq.P_Tsnot0}
\bbE[P] = (P_0f^3 + C) \!\! \left [ 1 - \frac{2 \! \left(1 \! - \! \frac{\lambda}{\mu f} \right )}{2\tau_s \! + \! \lambda \tau_s^2} \! \! \left( \!\! R' - \frac{1}{\mu f - \lambda} \right ) \! \right].
\end{align}
Thus we have an optimal $(\tau_c, f)$ pair (c.f.~Figure~\ref{fig.Tsnot0}). This suggests that one should not set $\tau_c$ and $f$ independently: they are coupled and depend on the quality of service requirement. 

The race-to-halt mechanism is a special case of this threshold mechanism with $\tau_c = 0$ and $f = 1$. That is, the server runs as fast as it could when the queue starts to build up and shuts down immediately after it clears all the jobs. The mean response time and power consumption reduce to:
\begin{align}
\label{eq.R_RaceHalt}
\bbE[R] &= \frac{1}{\mu - \lambda} + \frac{\tau_s}{2(1 + \lambda \tau_s)} + \frac{\tau_s}{2} \\
\label{eq.P_RaceHalt}
\bbE[P] &= (P_0 + C) \left(1 - \frac{1 - \frac{\lambda}{\mu}}{1 + \lambda \tau_s} \right ).
\end{align} 
The power consumption (\ref{eq.P_RaceHalt}) is a monotonically increasing function with respect to $\lambda$. However for mean response time, there is a $\lambda$ that minimizes the delay.



\subsection{Threshold mechanism with job batching}
\label{sec.singleserver_threshold_batching}
In this section we extend our analysis to consider the threshold mechanism with job batching. 
\begin{definition}[Threshold Mechanism with Job Batching]
\label{def.threshold_algorithm_job_batching}
This mechanism is the same as the one in Definition~\ref{def.threshold_mechanism} with the following difference. When the shut-down server sees the first job arrival, the server remains shut-down for some additional time $\tau_w$ before waking up. As before, wake-up takes times $\tau_s$.
\end{definition}

As we did for the basic threshold mechanism, in Figure~\ref{fig.threshold_batching_illustration} we provide a pair of sample paths illustrating the operation of the modified mechanism.  The upper sample path indicate queue occupancy.  The lower sample path is binary, indicating when the server is on and off.  Note the additional parameter vis-\`{a}-vis the basic mechanism.
\begin{figure}[t]
\centering
\includegraphics[scale=0.45]{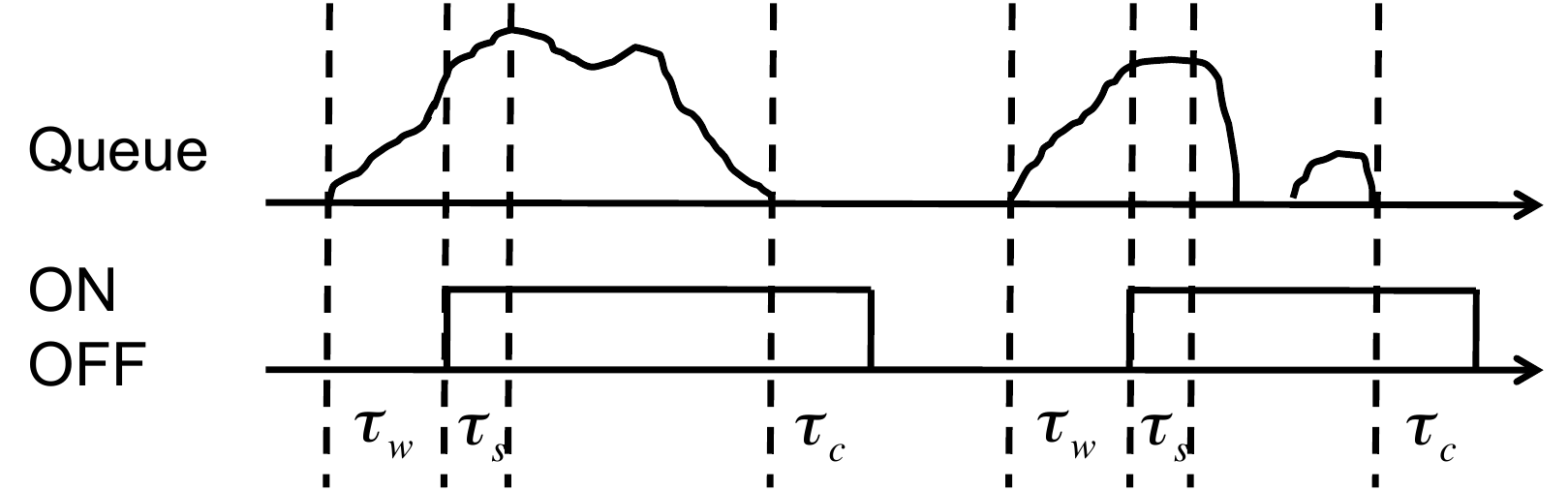}
\caption{Threshold mechanism with job batching.}
\label{fig.threshold_batching_illustration}
\end{figure}
The intuition behind this mechanism is that by batching more jobs at the beginning, it is less likely that the server will run out of jobs in the near future. This mechanism is the spirit in the periodic power-on and power-off operation in MapReduce clusters and the idea of batching database queries (see \cite{LangPatel} and the references therein). However, it is not clear how $\tau_w$ affects power and mean response time and the relation with $f$, $\tau_c$ and $\tau_s$ is unknown. We derive the mean response time and power consumption for this threshold mechanism with job batching in Lemma~\ref{thm.threshold_batching}.
\begin{lemma}
\label{thm.threshold_batching}
The mean response time and power consumption of the threshold mechanism with job batching are
\begin{align}
\label{eq.threshold_batching_R}
\bbE[R] &= \frac{1}{\mu f - \lambda} + \frac{2(\tau_s + \tau_w) + \lambda (\tau_s + \tau_w)^2 }{2(e^{\lambda \tau_c} + \lambda (\tau_s + \tau_w) ) } \\
\label{eq.threshold_batching_P}
\bbE[P] &= (P_0 f^3 + C)\left ( 1 - \frac{(1 + \lambda \tau_w) \left (1 - \frac{\lambda}{\mu f} \right )}{e^{\lambda \tau_c} + \lambda (\tau_s + \tau_w)} \right ).
\end{align}
\end{lemma}
\begin{proof}
The proof follows from the one in Theorem~\ref{thm.threshold}. In particular, the random delay $D$ now becomes $D = 0$ if $0 \leq T \leq \tau_c$ and $D = \tau_s + \tau_w$ if $T > \tau_c$. We obtain (\ref{eq.threshold_batching_R}) by solving for $\bbE[D]$ and $\bbE[D^2]$ and plugging in (\ref{eq.MG1_exceptional}) with $c_s^2 = 1$. The power consumption can also be derived in the same way as in Theorem~\ref{thm.threshold} with $f_{\rm off}$ replaced by:
\begin{align}
f_{\rm off} &= \frac{\int_{\tau_c}^{\infty} (t - \tau_c) \lambda e^{-\lambda t} dt + \int_{\tau_c}^{\infty} \tau_w \lambda e^{-\lambda t} dt}{L}.
\end{align} 
The rest of the proof follows from the one in Theorem~\ref{thm.threshold}.
\end{proof}
Note that when $\tau_w = 0$, the system reduces to the threshold mechanism. When $\tau_w$ is very large, the system waits long period of time before waking up: the mean response time thus goes unbounded and the power consumption converges to $\frac{\lambda}{\mu f} ( P_0 f^3 + C)$. 

Under a certain mean response time budget $\bbE[R] = R'$, there is an optimal triple $(\tau_c, f, \tau_w)$ that minimizes the power consumption. In particular, when $\tau_c = 0$, the mean response time and power consumption reduce to:
\begin{align}
\label{eq.threshold_batching_RtoH_R}
\bbE[R] &= \frac{1}{\mu f - \lambda} + \frac{2(\tau_s + \tau_w) + \lambda (\tau_s + \tau_w)^2 }{2(1 + \lambda (\tau_s + \tau_w) ) } \\
\label{eq.threshold_batching_RtoH_P}
\bbE[P] &= (P_0 f^3 + C)\left ( 1 - \frac{(1 + \lambda \tau_w) \left (1 - \frac{\lambda}{\mu f} \right )}{1 + \lambda (\tau_s + \tau_w)} \right ).
\end{align}
Further with $f = 1$, the threshold mechanism reduces to the race-to-halt mechanism with job batching.  We simulate its mean response time (\ref{eq.threshold_batching_RtoH_R}) and power consumption (\ref{eq.threshold_batching_RtoH_P}) in Figure~\ref{fig.BatchTc0}.  

\subsection{Simulation Results}
\label{sec.simulationresults}
In this section we present our simulation results for the fixed threshold mechanisms. We choose the simulation parameters in real data traces from many literatures (see \cite{MeisnerGold} and the references therein). 

\subsubsection{Threshold mechanism}
We consider a computing facility with $P_0 = 150$, $\mu = 1$ and $\lambda = 0.1$ which models low utilization scenario. If the wake-up cost is negligible, i.e., $\tau_s = 0$, then from previous analysis we have $\tau_c = 0$ and the mean response time and power consumption reduce to (\ref{eq.RP_Ts0}). Figure.~\ref{fig.Ts0Tc0} illustrates the power-delay tradeoff for various $C$ when $\tau_s = \tau_c = 0$. 
\begin{figure}[t]
\centering
\includegraphics[scale=0.48]{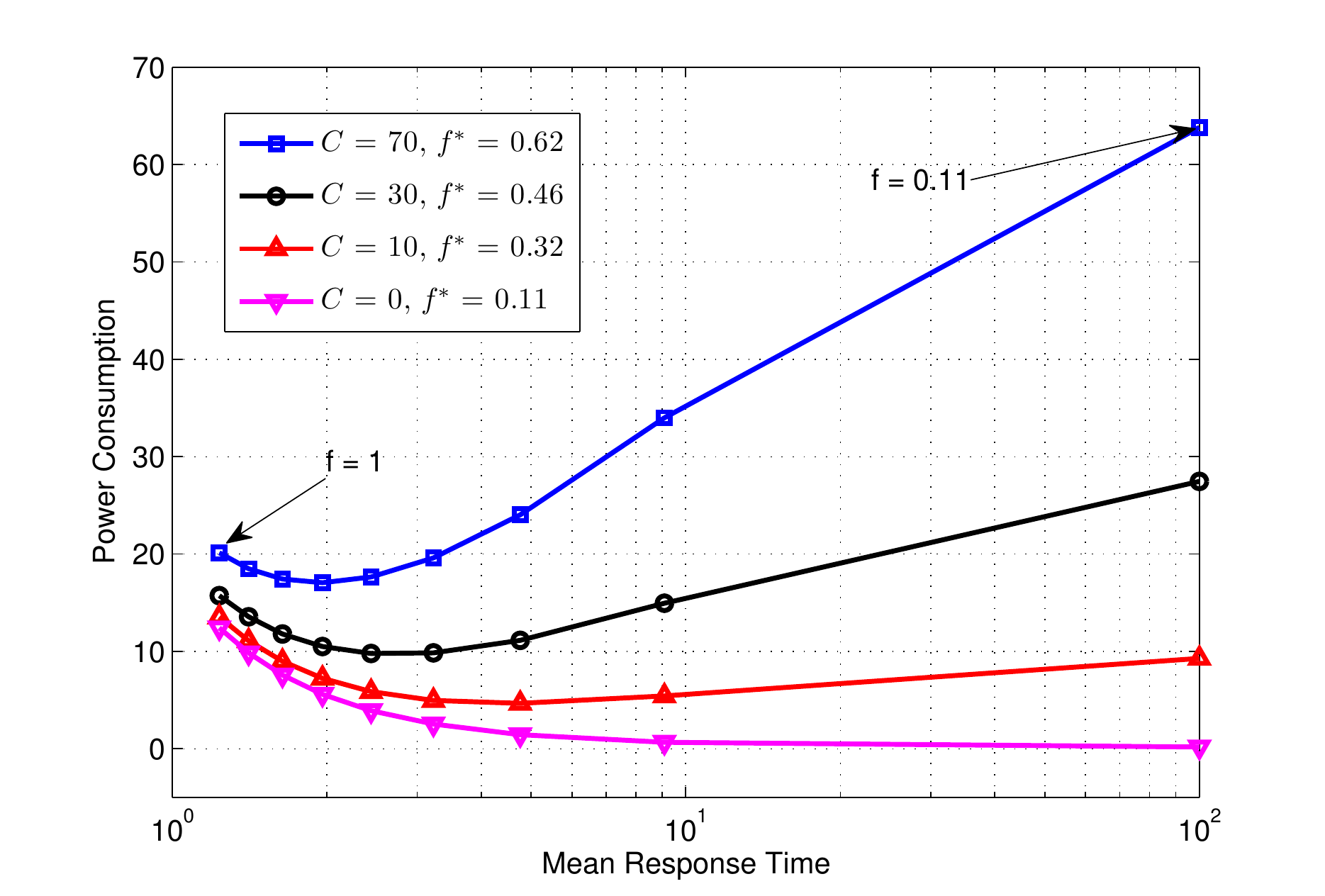}
\caption{Threshold mechanism, $\tau_s = \tau_c = 0$. Note that there is an optimal frequency $f^*$ that minimizes the power consumption.}
\label{fig.Ts0Tc0}
\end{figure}
Notice that there is an optimal frequency scaling that minimizes the power consumption. The results suggest that running jobs at large delay (using low frequency) may actually consume more power to run. 

In a more realistic scenario where $\tau_s \neq 0$, Figure~\ref{fig.Tsnot0} validates our argument that there is an optimal $(\tau_c, f)$ pair that jointly minimizes the power consumption given a target mean response time (c.f.~(\ref{eq.P_Tsnot0})). We set $P_0 = 150$, $C = 70$, $\lambda = 0.1$, $\mu = 1$ and $\tau_s = 10$.
\begin{figure}[t]
\centering
\includegraphics[scale=0.48]{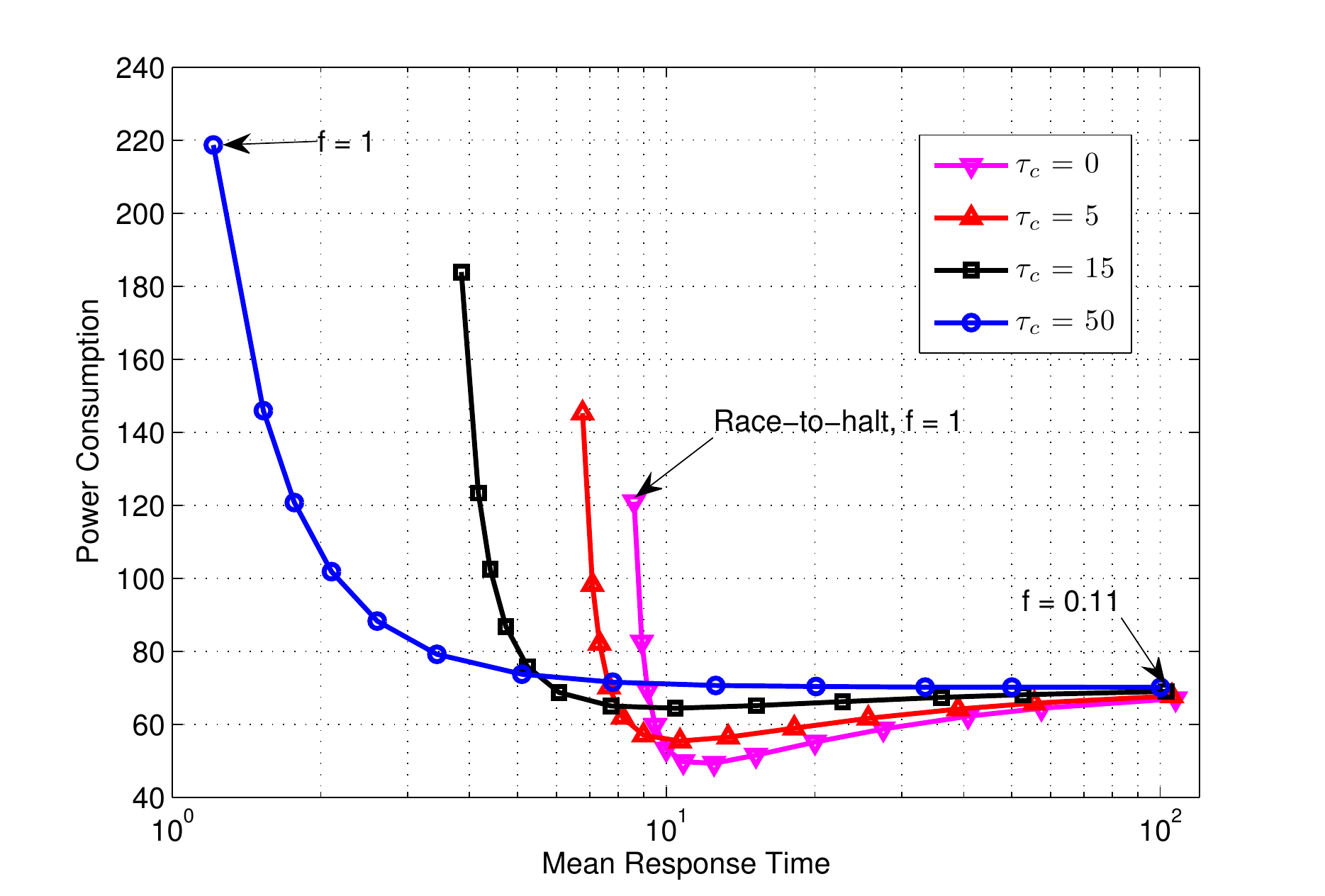}
\caption{Threshold mechanism, $\tau_s = 10$. Different target delay corresponds to different $\tau_c$ and $f$ pair.}
\label{fig.Tsnot0}
\end{figure}
Notice that for a given mean response time $R'$, there is an optimal $\tau_c$ and an associated frequency scaling $f$ that minimize the power consumption. Note also that for the mean response time achieved by the race-to-halt mechanism ($\tau_c = 0$, $f = 1$), we can pick another $(\tau_c, f)$ pair that yields smaller power consumption.

\subsubsection{Threshold mechanism with job batching}
We simulate the mean response time (\ref{eq.threshold_batching_RtoH_R}) and power (\ref{eq.threshold_batching_RtoH_P}) for the threshold mechanism with job batching. We set $P_0 = 150$, $C = 70$, $\lambda = 0.1$, $\mu = 1$, $\tau_s = 10$ and $\tau_c = 0$. The frequency scaling $f$ and batching period $\tau_w$ are kept as variables. The power-delay tradeoff is shown in Figure~\ref{fig.BatchTc0}.
\begin{figure}[t]
\centering
\includegraphics[scale=0.48]{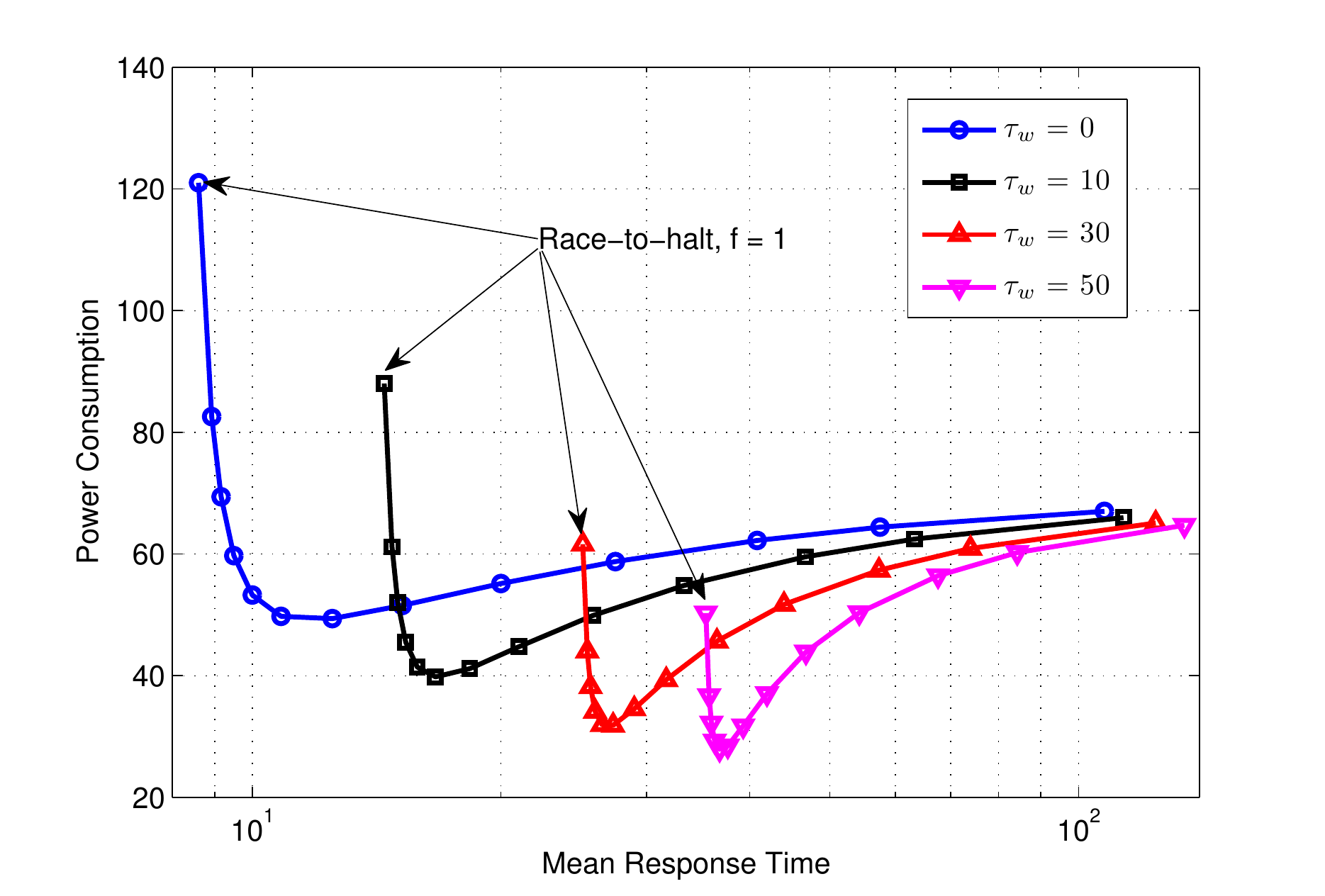}
\caption{Threshold mechanism with job batching, $\tau_s = 10$ and $\tau_c = 0$. Different target delay corresponds to different $\tau_w$ and $f$ pair. Note that the curve with $\tau_w = 0$ is the same as the one with $\tau_c = 0$ in Figure~\ref{fig.Tsnot0}.}
\label{fig.BatchTc0}
\end{figure}
Notice that for some mean response time achieved by the race-to-halt mechanism, we can pick another $(\tau_c, f)$ pair that yields smaller power consumption. The intuition is that to save power, one typically prefers smaller $f$ over $f = 1$. However to maintain the same delay performance one needs to compensate the increase in delay caused by the smaller $f$ by picking a smaller $\tau_w$. Meanwhile, one should not decrease $f$ too much either as doing so the peripheral power $C$ will soon be the dominating factor. We also note that the power-delay tradeoff is monotonic for race-to-halt scheme: increasing $\tau_w$ always incurs larger delay and lower power consumption.
\section{Multi-server Analysis}
\label{sec.multiserver}
In this section we extend our queuing analysis to study the interplay between frequency scaling and facility plant size, i.e., the number of servers. We study two simple multi-server scenarios, namely flow splitting and job splitting. We observe that even in such simple settings there are optimal operating frequency and plant size pairs that minimize the power consumption. 

Consider $n$ parallel homogeneous servers with a centralized job dispatcher. Jobs arrive at the dispatcher according to a Poisson process with rate $\lambda$. The job dispatcher distributes jobs to servers according to some rules. In this section we consider two simple rules: flow splitting using Bernoulli splitting and job splitting using fork-join. We assume all servers use the same operating frequency scaling $f$, each consuming $P_0 f^3 + C$ amount of power.

\subsection{Flow splitting}
In the flow splitting case, the job dispatcher sends jobs to servers according to a Bernoulli splitting manner. Each server behaves as an M/M/1 queue with Poisson arrival rate $\lambda/n$. 
\begin{lemma}
\label{thm.flowsplitting}
The mean response time and power consumption of flow splitting multi-server system are:
\begin{align}
\label{eq.flowsplit_PN}
\bbE[R] = \frac{1}{f \mu - \frac{\lambda}{n}} \quad \quad \bbE[P] = n (P_0 f^3 + C).
\end{align}
\end{lemma}

For any given $\bbE[R] = R'$, simple algebraic calculations show that there is an optimal frequency scaling and plant size pair that minimizes the power consumption. In particular, in large delay region $R' = \infty$, the optimal frequency scaling $f$ and plant size $n$ are given by:
\begin{align}
\label{eq.optimal_flowsplit_freqn}
f= \sqrt[3]{\frac{C}{2P_0}} \quad \quad n = \frac{\lambda}{\mu f}.
\end{align} 
This suggests that for power-efficient computation, it is not necessarily true that running as fast as possible or consolidating jobs onto as few servers as possible offers a better power efficiency. This phenomenon is visualized in Figure~\ref{fig.flowsplit}. We conjecture that similar observations exist for round robin scheduling where the inter-arrival time between jobs is Erlang-n distributed.

\subsection{Job splitting}
In the job splitting case, upon a job arrival the job dispatcher immediately makes $n$ copies of the job and forks them in parallel to $n$ servers. This models the queries to content retrieval databases where each incoming request can be simultaneously routed to $n$ databases waiting for some of them to respond.  Servers process requests in parallel and one queue is maintained at each server. When any $k$ out of $n$ servers respond, the rest of $n-k$ servers abandon the corresponding requests and the job departs the system. Such system is often termed $(n, k)$ fork-join queue \cite{NelsonTantawi} in queuing theory literature. There is no known close form solution for the mean response time of the fork-join system, not even for $(n, n)$ system. However several bounds exist (for example, see \cite{JoshiLiu}). For the job splitting case, working with the bounds we notice that there is also an optimal frequency scaling $f$ and plant size $n$ combination such that the power is minimized for a given delay budget. 

In both flow splitting and job splitting cases, packing jobs onto fewer servers requires faster processor speed to maintain a given delay performance thus increasing the processing power $P_0f^3$. On the other hand, provisioning more servers always incurs the fixed peripheral power expenditure $C$. 

\subsection{Simulation Results}
We simulate the mean response time and power consumption for multi-server flow splitting case. The case for job splitting shares the same spirit (omitted due to page limits). We set $P_0 = 150$, $C = 10$, $\lambda = 0.7$ and $\mu = 1$ while the frequency scaling $f$ and the number of servers $n$ are kept as variables. Simulation results are shown in Figure~\ref{fig.flowsplit}.
\begin{figure}[t]
\centering
\includegraphics[scale=0.48]{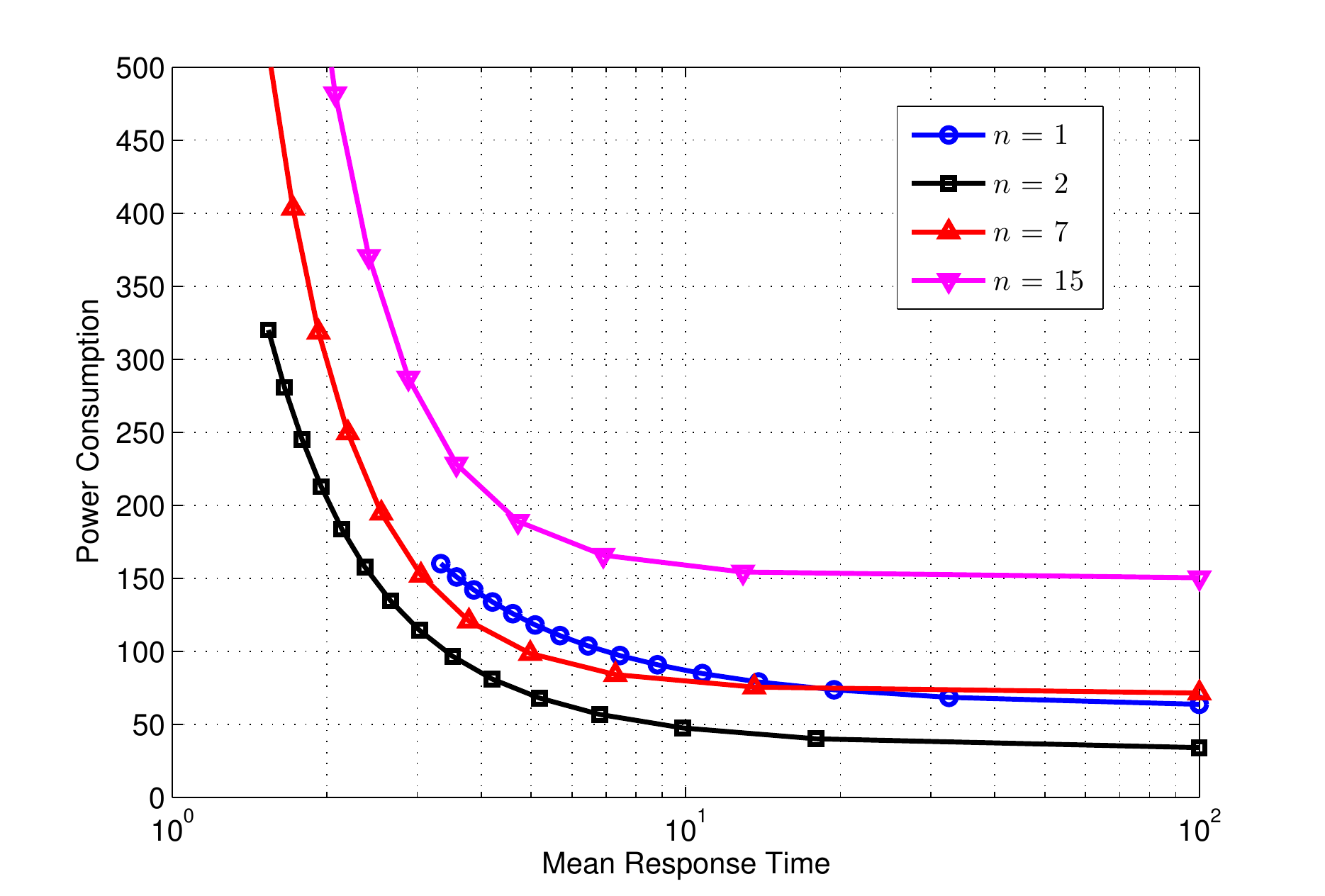}
\caption{Power delay tradeoff with flow splitting.}
\label{fig.flowsplit}
\end{figure}
For each $n$, we simulate different frequency scaling $f$ to plot the curve. Note that for some fixed mean response time, the power consumption first decreases then increases with increasing $n$. Intuitively, in one extreme case where jobs can tolerate large delay, the system should run slowly with small amount of servers (c.f.~(\ref{eq.optimal_flowsplit_freqn})). In another extreme case where jobs demand fast response, the system should run faster with many severs powered on. 
\section{Conclusion and Future Work}
\label{sec.conclusion}
In this paper we present a queuing theoretic analysis of some widely used power-efficient operations in modern computing. We analytically characterize the power-delay tradeoff for the threshold mechanisms with and without job batching. We also analyze the multi-server case. For these mechanisms we discover that there oftentimes exist sweet spots: optimal combinations of processing speed and other system parameters that yield best power efficiency. 

There are many promising future directions. These include the investigation of other power-efficient mechanisms. For the single server case, we will consider predictive wake-up and shut-down routines (c.f.~Figure~\ref{fig.category}). Such proactive control requires some prediction tools to predict traffic and offers improvements in delay. For the multi-server case,  we question the power efficiency of many conventional dispatching algorithms as most of them are not traditionally designed for power-efficient computing. We would like to understand the interplay and investigate the optimality between dispatching mechanisms and other system parameters. This will motivate some design guidances for power-efficient job dispatching routines.

\bibliographystyle{styles/IEEEtran}
\bibliography{myrefs}
\end{document}